\documentclass[screen, nonacm]{acmart}

\AtBeginDocument{%
  \providecommand\BibTeX{{%
    \normalfont B\kern-0.5em{\scshape i\kern-0.25em b}\kern-0.8em\TeX}}}


\usepackage{xspace}
\usepackage{subcaption}
\usepackage{todonotes}
\usepackage{mathtools}
\usepackage{amsmath}
\usepackage{dsfont}
\usepackage[makeroom]{cancel}
\usepackage{algorithm}
\usepackage{algorithmic}

\newtheorem{theorem}{Theorem}[section]
\newtheorem{corollary}[theorem]{Corollary}
\newtheorem{definition}[theorem]{Definition}
\newtheorem{lemma}[theorem]{Lemma}
\newtheorem{proposition}[theorem]{Proposition}
\newtheorem{remark}{Remark}
\newtheorem*{problem}{Problem}

\definecolor{myred}{RGB}{108,0,0}
\graphicspath{{./Figures/}}

\DeclareMathOperator*{\argmax}{arg\,max}

\newcommand{\real}{\mathbb{R}}

\newcommand{\caimdp}{\NoCaseChange{\textsc{ca}}\textsc{IMDP}\xspace}

\newcommand{\pr}{\mathbb{P}}

\newcommand{\expec}{\mathrm{E}}

\newcommand{\adv}{\boldsymbol{\xi}}
\newcommand{\Adv}{\Xi}

\newcommand{\ind}{\mathds{1}}

\newcommand{\Qimdp}{\mathcal{Q}}
\newcommand{\Aimdp}{\mathcal{A}}
\newcommand{\I}{\mathcal{I}}
\newcommand{\Pup}{\hat{P}}
\newcommand{\Plow}{\check{P}}

\newcommand{\Pathimdpfin}{\mathit{Paths}^{\mathrm{fin}}}

\newcommand{\Str}{\Pi}
\newcommand{\str}{\boldsymbol{\pi}}
\newcommand{\ver}{\mathrm{ver}}

\copyrightyear{2023} 
\acmYear{2023} 
\setcopyright{rightsretained} 
\begin{document}
\title{Interval Markov Decision Processes with Continuous Action-Spaces}

\author{Giannis Delimpaltadakis}
\email{i.delimpaltadakis@tue.nl}
\affiliation{%
     \department{Control Systems Technology Group, Mechanical Engineering}
      \institution{Eindhoven University of Technology}
      \city{Eindhoven}
      \country{Netherlands}
    }
\author{Morteza Lahijanian}
\email{morteza.lahijanian@colorado.edu}
\affiliation{%
    \department{Aerospace Engineering Sciences \& Computer Science}
    \institution{University of Colorado Boulder}
    \city{Boulder, Colorado}
    \country{U.S.A.}
}
\author{Manuel Mazo Jr.}
\email{m.mazo@tudelft.nl}
\affiliation{%
    \department{Delft Center for Systems and Control, Mechanical Engineering}
    \institution{Delft University of Technology}
    \city{Delft}
    \country{Netherlands}
}
\author{Luca Laurenti}
\email{l.laurenti@tudelft.nl}
\affiliation{%
    \department{Delft Center for Systems and Control, Mechanical Engineering}
    \institution{Delft University of Technology}
    \city{Delft}
    \country{Netherlands}
}

\begin{abstract}
Interval Markov Decision Processes (IMDPs) are finite-state uncertain Markov models, where the transition probabilities belong to intervals. Recently, there has been a surge of research on employing IMDPs as abstractions of stochastic systems for control synthesis. However, due to the absence of algorithms for synthesis over IMDPs with continuous action-spaces, the action-space is assumed discrete a-priori, which is a restrictive assumption for many applications. Motivated by this, we introduce continuous-action IMDPs ({\caimdp}s), where the bounds on transition probabilities are functions of the action variables, and study value iteration for maximizing expected cumulative rewards. Specifically, we decompose the max-min problem associated to value iteration to $|\Qimdp|$ max problems, where $|\Qimdp|$ is the number of states of the {\caimdp}. Then, exploiting the simple form of these max problems, we identify cases where value iteration over {\caimdp}s can be \emph{solved efficiently} (e.g., with linear or convex programming).
We also gain other interesting insights: e.g., in certain cases where the action set $\Aimdp$ is a polytope, 
synthesis over a discrete-action IMDP, where the actions are the vertices of $\Aimdp$, is sufficient for optimality. We demonstrate our results on a numerical example. Finally, we include a short discussion on employing {\caimdp}s as abstractions for control synthesis.
\end{abstract}

\maketitle

\section{Introduction}
\subsection{Motivation and Contributions}
Interval Markov Decision Processes (IMDPs; alternatively called bounded-parameter Markov decision processes) are a class of uncertain finite-state Markov Decision Processes (MDPs), where transition probabilities between states are only known to belong to intervals \citep{givan2000bounded}. Due to their modelling flexibility and the availability of efficient planning algorithms \citep{givan2000bounded}, IMDPs have recently gained popularity in the control and computer-science communities for verification of and control-synthesis for uncertain systems (see, e.g., \citep{lahijanian2015formal,laurenti2020formal,coogan2022abstraction, lavaei2022automated,delimpaltadakis2022formal}). However, due to the absence of computational algorithms for IMDPs with continuous action-spaces, the action set is, generally, assumed discrete \citep{lahijanian2015formal,laurenti2020formal}. This is a restrictive assumption as many realistic control applications involve continuous underlying action-spaces. Furthermore, with the current available algorithms, the only way to address continuous action-spaces is either a) to discretize the action-space and deal with discrete-action IMDPs, or b) solve the optimization problems associated to control synthesis using heuristics, as recently proposed in \cite{coogan2022abstraction}. Nonetheless, a) scales exponentially with the dimension of the action space, rendering high-dimensional problems intractable. Moreover, b) - and a), when the discretization is carried out blindly - results into suboptimal solutions without suboptimality bounds.

Motivated by the above, we introduce continuous-action Interval Markov Decision Processes ({\caimdp}s), where the transition probability intervals are functions of the action variables. We study value iteration over {\caimdp}s for synthesizing policies that maximize finite-horizon pessimistic expected cumulative rewards (which is referred to as the \emph{robust control problem}), and identify cases where value iteration can be \emph{solved efficiently}. Specifically, we show that the max-min optimization problem associated to value iteration is equivalent to solving $|\Qimdp|$ maximization problems, where $|\Qimdp|$ is the number of states of the {\caimdp}. Then, exploiting the simple form of these maximization problems, we distinguish the following tractable cases: 1) linear (on the action variable) bounds on transition probabilities and polytopic action set $\Aimdp$, where the maximization problems are linear programs, 2) concave and convex, respectively, transition bounds and convex $\Aimdp$, where we have convex programs, and 3) convex and concave, respectively, transition bounds and polytopic $\Aimdp$, which amounts to convex maximization over polytopes. \emph{As a result, in these cases, control synthesis over caIMDPs, via value iteration, comes with guaranteed optimality and fast computation.} Moreover, we gain further interesting insights. For example, in cases 1 and 3, synthesizing over discrete-action IMDPs, where the action set consists of the vertices of $\Aimdp$, is sufficient for optimality. Overall, these results provide the necessary theory to synthesize optimal policies efficiently over {\caimdp}s. 

We showcase our results on a numerical example, which demonstrates that synthesizing over a \caimdp is not only the only way to optimality, but it is also computationally efficient, even when compared to synthesizing over discrete-action IMDPs. Finally, we briefly discuss about using {\caimdp}-abstractions for control synthesis for stochastic systems, touching upon constructing the abstraction and obtaining suboptimality\footnote{
The policy computed over the (ca)IMDP-abstraction, even if optimal for the abstraction, will generally be suboptimal for the original system. That is because the continuous state-space has been partitioned to a finite number of subsets, which represent the (ca)IMDP's states.
} bounds on the obtained policy.

\subsection{Related work}
While several existing works have focused on developing control and verification algorithms for finite-state and finite-action IMDPs  \cite{givan2000bounded,koutsoukos2006computational,nilim2005robust}, efficient algorithms for control of IMDPs with continuous action-spaces are missing. A first attempt to close this gap has been proposed in \cite{coogan2022abstraction}, where the authors present an algorithm to synthesize policies that maximize a reachability probability based on suboptimal heuristics and non-convex optimization. In contrast, here we derive an exact reformulation of value iteration for {\caimdp}s, which leads to tractable solutions based on linear or convex programming, in many cases of interest. {\caimdp}s are also  closely related to parametric MDPs, which are Markov decision processes with finite state-/action-spaces, but where the transition probabilities may depend on some parameters \citep{hahn2011probabilistic,lanotte2007parametric,cubuktepe2021convex}. However, a key difference is that, in parametric MDPs, the problem is to find parameter values that are fixed and independent of the state of the system. Instead, in {\caimdp}s we seek feedback control policies that depend on both time and state, thus requiring substantially different approaches. 

\subsection{Notation}
$\real$ (resp., $\real_{\geq0}$) stands for the set of real numbers (resp., non-negative reals). Given a polytope $\Aimdp\subset\real^{n}$,
its set of vertices is denoted by $\ver(\Aimdp)$. Given a discrete set $\Qimdp$, we denote its cardinality by $|\Qimdp|$. By $\ind$ we denote vectors, of appropriate dimension, of which all entries are equal to $1$. For a vector $x\in\real^n$, $x\succeq0$ ($x\preceq0$) denotes that all its components are non-negative (non-positive, resp.). 

\section{Continuous-Action Interval Markov Decision Processes (\NoCaseChange{\textsc{ca}}\textsc{IMDP}\NoCaseChange{s})}
\subsection{The basic elements of a \caimdp}
Continuous-action interval Markov decision processes ({\caimdp}s) generalize finite-state MDPs with interval-valued transition probabilities and continuous action-spaces. 

\begin{definition}[\caimdp] \label{def:imdp}
    A continuous-action interval Markov decision process (\caimdp) is a tuple $\I = (\Qimdp,\Aimdp,\Plow,\Pup,R)$, where
    \begin{itemize}
    	\setlength\itemsep{1mm}
    	\item $\Qimdp$ is the finite set of states,
    	\item $\Aimdp\subset \real^{n_{\Aimdp}}$ is the
    	set of actions,
        \item $\Plow: \Qimdp \times \Aimdp \times \Qimdp \rightarrow [0,1]$, where $\Plow(q,a,q')$ is the lower bound on the transition probability from state $q \in \Qimdp$ to state $q' \in \Qimdp$ under action $a \in \Aimdp$,
        \item $\Pup: \Qimdp \times \Aimdp \times \Qimdp \rightarrow [0,1]$, where $\Pup(q,a,q')$ is  the upper bound on the transition probability from state $q \in \Qimdp$ to state $q' \in \Qimdp$ under action $a \in \Aimdp$,
        \item $R:  \Qimdp \to \real_{\geq0}$ is a bounded state-dependent reward function.  
    \end{itemize}
\end{definition}
For all $q,q' \in \Qimdp$ and $a \in \Aimdp$, it holds that $\Plow(q,a,q') \leq \Pup(q,a,q')$ and $\sum_{q' \in \Qimdp} \Plow(q,a,q') \leq 1 \leq \sum_{q' \in \Qimdp} \Pup(q,a,q')$. Given a state $q\in \Qimdp$ and an action $a\in\Aimdp$, a transition probability distribution $p_{q,a}:\Qimdp\to[0,1]$ is called \textit{feasible} if $\check{P}(q,a,q')\leq p_{q,a}(q')\leq\hat{P}(q,a,q')$ for all $q'\in \Qimdp$. The set of all feasible distributions for the state-action pair $(q,a)$ is denoted by $\Gamma_{q,a}$. We define $\Gamma = \{\Gamma_{q,a}:(q,a)\in\Qimdp\times\Aimdp\}$ to be the set of all feasible distributions for all state-action pairs. 
    
A path of a \caimdp is a sequence of states and actions $\omega = (q_0,a_0),(q_1,a_1),\dots$, where $(q_i,a_i)\in\Qimdp\times\Aimdp$, and we denote $\omega(k) = (q_k,a_k)$ and $\omega^q(k)=q_k$, for $k=0,1,\dots$ We denote the set of all finite paths by $\Pathimdpfin$. For a path $\omega \in \Pathimdpfin$, we denote its last element by $\omega(end)$ (i.e., for an $(N+1)$-length path $\omega(end)=\omega(N)=(q_{N},a_{N})$). 

To describe the evolution of a {\caimdp}-path, we also need to introduce the notions of \emph{policy} and \emph{adversary}.
\begin{definition}[policy]
\label{def:policy}
    For a \caimdp, a policy $\str: \Pathimdpfin\times\Qimdp \rightarrow \Aimdp$ is
    a function that, given a finite path $\omega\in\Pathimdpfin$ and a state $q\in\Qimdp$, returns an action. The set of all policies is denoted by $\Str$. A policy is Markov, if the choice of action depends only on $q$ and on the path's length. 
\end{definition}

\begin{definition}[Adversary]
\label{def:adversary}
For a \caimdp, an adversary is a function $\adv: \Pathimdpfin \rightarrow \Gamma$. Given a finite path $\omega \in \Pathimdpfin$, it returns a feasible distribution $p_{q,a}\in\Gamma_{q,a}$, where $(q,a)=\omega(end)$. The set of all adversaries is denoted by $\Adv$. An adversary is Markov, if the choice of a feasible distribution depends only on $\omega(end)$ and the path's length. 
\end{definition}
Given a policy $\str$ and an adversary $\adv$, a {\caimdp}-path evolves as follows. At time $i$, given the finite path $(q_0,a_0),\dots,$ $(q_{i-1},a_{i-1})$ and the current state $q_i$, the policy $\str$ chooses the action $a_i$. Then, the adversary $\adv$, given the path $(q_0,a_0),\dots(q_{i},a_{i})$, chooses a feasible distribution $p_{q_i,a_i}\in\Gamma_{q_i,a_i}$. The next state of the path $q_{i+1}$ is sampled randomly from $p_{q_i,a_i}$. 

In words, a policy is a control strategy, that, at each time, decides the control action based on the history of the path. The adversary represents the environment: once the control action is taken, the adversary resolves the uncertainty stemming from the transition probability intervals, thus fixing the -still stochastic- environment for the current time step. Notice that, given a specific policy $\str$ and adversary $\adv$, the \caimdp collapses to a time-varying Markov chain. Thus, given $\str$, $\adv$, an initial state $q_0$, and a \emph{horizon} $N$, a probability measure is uniquely defined over $(N+1)$-length paths $\omega\in\Pathimdpfin$ \cite{bertsekas2004stochastic}.

\begin{remark}\label{rem:deterministic policies}
For ease of presentation, we constrain the definitions of policy and adversary to deterministic ones. In fact, they can be random, but, as Proposition \ref{prop:caimdps_value function} shows, the optimal rewards that we study are achieved by deterministic, Markov policies and adversaries; thus, it suffices to consider only deterministic ones.
\end{remark}
\subsection{Optimal policies, optimal rewards and value iteration}
In what follows, we implicitly assume that all mentioned $\max\min$ quantities (e.g., the ones in \eqref{eq:robust_control_problem} and \eqref{eq:original_vi2}) are well-defined. Numerous sets of assumption can be employed to impose this: e.g., as a consequence of Proposition \ref{prop:caimdps_value function} below, it suffices that $\Aimdp$ is compact and $\Plow,\Pup$ are continuous functions of the action variables $a$. 

Given $\str$, $\adv$, $N$, an initial state $q_0  \in \Qimdp$, and a factor $\gamma\geq 0$, we define the so-called \emph{expected cumulative reward}:
\begin{equation*}
    \mathcal{R}^{N}_{\str,\adv}(q_0)=\expec^{q_0}\left[\sum_{i=0}^N \gamma^i R(\omega^q(i)) \mid \str,\adv \right]
\end{equation*}
where the expectation is taken w.r.t. the probability measure over $(N+1)$-length paths in $\Pathimdpfin$ starting from state $q_0$. Notice that $R$, as defined in Definition \ref{def:imdp}, depends only on states, not actions.

In this work, as done commonly in the literature \citep{givan2000bounded}, we consider the \emph{robust control problem}: finding a policy that maximizes the expected cumulative reward generated by the worst possible adversary (and the policy itself). In other words, we study the following max-min problem\footnote{While this work focuses on the max-min problem, max-max, min-max, and min-min problems can be addressed similarly.}:
\begin{problem}[Robust Control]
Given a \caimdp $\I$, a factor $\gamma\geq 0$, a horizon $N$, and an initial state $q_0$, solve:
\begin{equation}\label{eq:robust_control_problem}
   \mathcal{R}^{N}_{\star}(q_0) \equiv \max_{\str\in\Str}\min_{\adv\in\Adv}\mathcal{R}^{N}_{\str,\adv}(q_0).
\end{equation}
\end{problem}
A policy $\str_\star$ that solves \eqref{eq:robust_control_problem}, i.e., 
\begin{equation*}
 \str_\star\in\argmax_{\str\in\Str}\min_{\adv\in\Adv}\mathcal{R}^{N}_{\str,\adv}(q_0),
\end{equation*} is called a \emph{(pessimistically)}\footnote{"Pessimistically", as it optimizes the reward generated by the \emph{worst} adversary. If we considered the \emph{best} adversary instead, which would amount to a max-max problem, the policy would be called \emph{optimistically} optimal.} \emph{optimal policy}. The optimal value $\mathcal{R}^{N}_{\star}(q_0)$ of \eqref{eq:robust_control_problem} is called \emph{(pessimistic) optimal (expected cumulative) reward}. Solving \eqref{eq:robust_control_problem} is computing an optimal policy $\str_\star$ and the optimal reward.

For MDPs and for IMDPs with discrete action-spaces, it is well-known that 
\eqref{eq:robust_control_problem} can be solved by an iterative scheme called \emph{value iteration}. The following proposition extends this to {\caimdp}s as well:
\begin{proposition}\label{prop:caimdps_value function}
Consider a \caimdp $\I$. 
For any $j=0,\dots,N$, the following holds:
\begin{equation}
    \mathcal{R}^{j}_\star(q_0) = V_{N-j}(q_0)
\end{equation}
where $V_{N-j}(q)$ is defined, for any $q\in\Qimdp$, through the following iteration:
\begin{subequations}\label{eq:original_vi}
\begin{align}
    &V_{N}(q) = R(q)\label{eq:original_vi1}\\
    &V_{k-1}(q) = R(q) + \gamma\max_{a\in \Aimdp}\min_{p\in\Gamma_{q,a}}\sum_{q'\in\Qimdp}p(q')V_{k}(q')\label{eq:original_vi2}
\end{align}
for $k=N,\dots,1$.
\end{subequations}
\end{proposition}
\begin{proof}
See Section \ref{sec:appendix_proof_value_iteration}.
\end{proof}
The process of solving \eqref{eq:original_vi} for all iterations is called \emph{value iteration} and the obtained function $V_{0}(\cdot)$ is called \emph{value function}. A direct corollary of Proposition \ref{prop:caimdps_value function}, is that there exist Markov policies (and adversaries) achieving the optimal reward, defined as follows:
\begin{corollary}[to Proposition \eqref{prop:caimdps_value function}]\label{cor:markov_policies}
    Any Markov policy satisfying:
    $$\str_{\star}(q,N-k-1) \in \argmax_{a\in\Aimdp}\min_{p\in\Gamma_{q,a}}\sum_{q'\in\Qimdp}p(q')V_{k}(q')$$ is optimal, where $\str_{\star}(q,i)$ denotes the action taken if the current state is $q$ and the current time is $i$.
\end{corollary}
\begin{remark}
Note that, while in this work we only consider cumulative rewards, similarly to IMDPs with discrete action-spaces, our results can easily be extended to more general properties, such as bounded-time fragments of PCTL \cite{lahijanian2015formal} or LTL \cite{jackson2021strategy}. Infinite-horizon (unbounded-time) extensions are also possible, but particular care is needed to guarantee convergence of the value iteration of Proposition \ref{prop:caimdps_value function} due to the non-finiteness of $\Aimdp$.
\end{remark}

\section{Efficient formulations of value iteration over \NoCaseChange{ca}IMDP\NoCaseChange{s}}
By Proposition \ref{prop:caimdps_value function} and Corollary \ref{cor:markov_policies}, we see that to compute the optimal reward and an optimal policy it suffices to solve 
\eqref{eq:original_vi} for all iterations $k=N,\dots,1$. However, solving \eqref{eq:original_vi2} is not straightforward, as it includes solving (for each $q$) the max-min optimization problem $\max_{a \in A} \min_{p \in \Gamma_{q,a}} p^TV_{k}$, where we, abusively, denote $\allowdisplaybreaks{V_{k} = \begin{bmatrix}V_{k}(q_1) &\dots &V_{k}(q_{|\Qimdp|})\end{bmatrix}^\top}$, $p=\begin{bmatrix}p(q_1) &\dots &p(q_{|\Qimdp|})\end{bmatrix}^\top$. This max-min problem can be written as:
\begin{equation}\label{original_opti_problem}\tag{Mm}
\begin{aligned}
	\max_{a \in \Aimdp} \min_{p \in \real^{|\Qimdp|}}  \quad & p^\top V_{k}\\
	\mathrm{s.t.:} \quad & \check{P}(q,a,q_i)-p_i \leq 0, \text{ }i=1,2,\dots,|\Qimdp|,\\
				  & -\hat{P}(q,a,q_i)+p_i \leq 0, \text{ }i=1,2,\dots,|\Qimdp|,\\
				  & \sum_{i=1}^{|\Qimdp|}p_i -1 = 0\\
\end{aligned}
\end{equation}
where $p_i$ denotes the $i$-th component of vector $p$, and the constraints
guarantee that $p$ is a probability distribution (equality constraint) and belongs in $\Gamma_{q,a}$ (inequality constraints). 

In this section, we decompose problem \eqref{original_opti_problem} to a number of simpler problems, and identify special cases of $\Aimdp$, $\Plow(\cdot,a,\cdot)$ and $\Pup(\cdot, a, \cdot)$ in which problem \eqref{original_opti_problem} can be solved \emph{efficiently}. 
First, we state and prove this paper's main result (Theorem \ref{thm:main_theorem}): that problem \eqref{original_opti_problem} is equivalent to solving the $|\Qimdp|$ simpler maximization problems in \eqref{eq:final_LP}. This gives rise to Algorithm \ref{alg:algorithm}, which performs value iteration over {\caimdp}s, where, instead of solving \eqref{original_opti_problem}, it solves \eqref{eq:final_LP}. Then, the simple form of the maximization problems in \eqref{eq:final_LP} allows us to identify special cases in which they, and thus value iteration, are tractable. Furthermore, we show that when $\Aimdp$ is polytopic and $\Plow(\cdot,a,\cdot)$ and $\Pup(\cdot, a, \cdot)$ are either linear or convex and concave, respectively, performing synthesis over a discrete-action IMDP with its action set being $\ver(\Aimdp)$ is sufficient for optimality.

\subsection{Decomposing the max-min problem to $|\Qimdp|$ max problems}
To prove the main result, first we transform the original max-min problem \eqref{original_opti_problem} to a maximization problem, by employing duality:
\begin{proposition}\label{prop:dual_problem}
Consider the optimization problem:
\begin{equation}\label{eq:dual_opti_problem}
	\begin{aligned}
		\max_{a, \lambda_L, \lambda_U, \nu} \quad &\sum_{i=1}^{|\Qimdp|}\lambda_{L_i}\check{P}(q,a,q_i) - \sum_{i=1}^{|\Qimdp|}\lambda_{U_i}\hat{P}(q,a,q_i) -\nu\\
		\mathrm{s.t.:} \quad & a\in\Aimdp, \text{ }\lambda_L \succeq 0, \text{ }\lambda_U \succeq 0\\
		& V_{k}-\lambda_L + \lambda_U + \nu\ind=0
	\end{aligned}
\end{equation}
where $\lambda_L\in\real^{|\Qimdp|}, \lambda_U\in\real^{|\Qimdp|}$, $\nu\in\real$ and $\lambda_{L_i},\lambda_{U_i}$ denote the $i$-th component of $\lambda_L$ and $\lambda_U$, respectively. If $(a_\star,\lambda_{L_\star}, \lambda_{U_\star}, \nu_\star)$ solves \eqref{eq:dual_opti_problem}, then it also solves \eqref{original_opti_problem}. Moreover, the optimal values of \eqref{original_opti_problem} and \eqref{eq:dual_opti_problem} coincide.
\end{proposition}
\begin{proof}
See Section \ref{sec:appendix_proof_duality_prop}.
\end{proof}
The variables $\lambda_L$, $\lambda_U$ and $\nu$ are the so-called \emph{Lagrange multipliers}. Generally, even though just a max problem instead of a max-min problem, optimization problem \eqref{eq:dual_opti_problem} is not easy to work with directly. That is because of the product terms $\lambda_{L_i}\check{P}(q,a,q_i)$ and $\lambda_{U_i}\hat{P}(q,a,q_i)$ (e.g., even when $\Plow(\cdot,a,\cdot)$ and $\Pup(\cdot, a, \cdot)$ are linear on $a$, problem \eqref{eq:dual_opti_problem} is, generally, nonconvex). Nonetheless, by working around its special structure, we are able to dissect it into $|\Qimdp|$ simpler max problems, which constitutes this work's main result:
\begin{theorem}\label{thm:main_theorem}
Assume, without loss of generality, that $V_{k}(q_i)\geq V_{k}(q_{i+1})$ for all $i=1,2,\dots,|\Qimdp|-1$, i.e., that $V_{k}$ is in descending order. Consider the following optimization problem:
\begin{equation}\label{eq:final_LP} \tag{MP}
	\begin{aligned}
		\max_{j=1,2,\dots,|\Qimdp|}\max_{a} \quad &\sum_{i=1}^{j-1}\Big(V_{k}(q_i)-V_{k}(q_{j})\Big)\check{P}(q,a,q_i) + \sum_{i=j+1}^{|\Qimdp|}\Big(V_{k}(q_i)-V_{k}(q_{j})\Big)\hat{P}(q,a,q_i) +V_{k}(q_{j})\\
		\mathrm{s.t.:} \quad &a\in\Aimdp
	\end{aligned}
\end{equation}
If $a_\star$ solves \eqref{eq:final_LP}, then it solves \eqref{original_opti_problem}. Moreover, the optimal values of \eqref{eq:final_LP} and \eqref{original_opti_problem} coincide.
\end{theorem}
\begin{proof}
See Section \ref{sec:appendix_proof_main}.
\end{proof}

By employing Theorem \ref{thm:main_theorem}, we can devise an algorithm (Algorithm \ref{alg:algorithm}) to solve value iteration over {\caimdp}s, where, instead of solving the stringent \eqref{original_opti_problem}, we solve \eqref{eq:final_LP}, which can be solved efficiently in many cases of interest, as it is revealed in the next section. Observe how in each iteration of value iteration, we sort the value-function vector in decreasing order, thus imposing the assumption of Theorem \ref{thm:main_theorem} (hence why it is written that the assumption is without loss of generality).
\begin{algorithm}
\caption{Value Iteration on {\caimdp}s}\label{alg:algorithm}
    \begin{algorithmic}[1]
        \FOR{$q\in\Qimdp$} 
        \STATE $V_N(q)\gets R(q)$
        \ENDFOR
        \FOR{$k=N,\dots,1$}
        \STATE  $[V_{sorted},indices_{sorted}] = sort(V_k)$ \COMMENT{sorts $V_k$ in descending order and stores it in $V_{sorted}$. $indices_{sorted}$ describes the arrangement of the elements of $V_k$ into $V_{sorted}$, i.e. $V_k = V_{sorted}(indices_{sorted})$.}
        \FOR{$q$ in $indices_{sorted}$}
        \STATE $V_{k-1}(q) = R(q) + \gamma$*Solve(MP) 	\COMMENT{where in \eqref{eq:final_LP} we replace $V_k$ with $V_{sorted}$.} 
        \ENDFOR
        \ENDFOR
	\end{algorithmic}
\end{algorithm}

It remains to deduce when is it easy to perform line 7 of Algorithm \ref{alg:algorithm}, i.e. to solve \eqref{eq:final_LP}. Notice how \eqref{eq:final_LP} is in a much simpler form to analyze than the primal max-min problem \eqref{original_opti_problem}: a) they are just max problems, b) their only decision variable is the action $a$ and the constraint set is $\Aimdp$, and c) the functions $\Plow(q,a,q_i)$ and $\Pup(q, a, q_i)$ appear linearly, which facilitates convexity arguments. It is this simple form that allows us to deduce when they can be solved efficiently, in the following section.
\begin{remark}\label{rem:similarity to givan}
Notice that the derived maximization problems \eqref{eq:final_LP} and value-iteration algorithm share similarities with the algorithm proposed in \cite{givan2000bounded} to solve value iteration over discrete-action IMDPs. First, they both sort the value-function vector in each iteration. Furthermore, in both algorithms, given a specific action $a$, a combination of additions of lower and upper bounds on transition probabilities is found, which is optimal w.r.t. the ordered value-function vector. The number of possible such combinations is $|\Qimdp|$ in both cases. Nevertheless, contrary to here, in \cite{givan2000bounded} there is a stopping criterion, based on which the algorithm may terminate before computing all $|\Qimdp|$ combinations. The existence of such a stopping criterion for problems \eqref{eq:final_LP} (i.e., for the continuous-action case) is an open question, and a potential affirmative answer would significantly reduce computations.
\end{remark}

\subsection{When can value iteration for {\caimdp}s be solved efficiently?}\label{sec:cases}
By inspecting \eqref{eq:final_LP}, we are able to identify cases where value iteration for {\caimdp}s can be solved efficiently.

\textit{\textbf{The linear case.}} In the case of polytopic $\Aimdp$, and linear $\Plow(\cdot,a,\cdot)$, $\Pup(\cdot, a, \cdot)$, solving problem \eqref{eq:final_LP} amounts to solving $|\Qimdp|$ linear programs (LPs):

\begin{corollary}[to Theorem \ref{thm:main_theorem}]\label{cor:LP1}
If $\Aimdp$ is a polytope and both $\Plow(\cdot,a,\cdot)$ and $\Pup(\cdot, a, \cdot)$ are linear on $a$, then the maximization problems in \eqref{eq:final_LP} are LPs. 
\end{corollary}
Moreover, notice that an optimal $a_\star$ lies on the vertices of $\Aimdp$, since the problems in \eqref{eq:final_LP} are LPs with $\Aimdp$ as the constraint set. Thus, another corollary of Theorem \ref{thm:main_theorem} is that, instead of solving \eqref{eq:final_LP} in line 7 of Algorithm \ref{alg:algorithm}, we could solve the original max-min problem \eqref{original_opti_problem} only for $a\in\ver(A)$, which amounts to solving $|\ver(\Aimdp)|$ LPs:
\begin{corollary}[to Theorem \ref{thm:main_theorem}]\label{cor:LP2}
If $\Aimdp$ is a polytope and both $\Plow(\cdot,a,\cdot)$ and $\Pup(\cdot, a, \cdot)$ are linear on $a$, solving the max-min problem \eqref{original_opti_problem} amounts to solving the following $|\ver(\Aimdp)|$ minimization LPs: 
\begin{equation}\label{eq:vertices}
	\begin{aligned}
		\max_{a\in\ver(\Aimdp)}\min_{p \in \real^{|\Qimdp|}}  \quad & p^\top V_{k}\\
	\mathrm{s.t.:} \quad &\check{P}(q,a,q_i)-p_i \leq 0, \text{ }i=1,2,\dots,|\Qimdp|,\\
				  &-\hat{P}(q,a,q_i)+p_i \leq 0, \text{ }i=1,2,\dots,|\Qimdp|,\\
				  &\sum_{i=1}^{|\Qimdp|}p_i -1 = 0\\
	\end{aligned}
\end{equation}
\end{corollary}
\begin{remark}\label{rem:discreteIMDP_linear}
Equivalently, Corollary \ref{cor:LP2} guarantees that it is sufficient to construct a discrete-action IMDP with its action set being $\ver(\Aimdp)$, and solve the synthesis problem over the discrete-action IMDP.
\end{remark}
Corollary \ref{cor:LP1} amounts to $|\Qimdp|$ LPs with $n_{\Aimdp}$ decision variables and $\Aimdp$ as a constraint set. On the other hand, the approach of Corollary \ref{cor:LP2} (or Remark \ref{rem:discreteIMDP_linear}) amounts to $|\ver(\Aimdp)|$ LPs with $|\Qimdp|$ decision variables. Note, however, that to solve the LPs of Corollary \ref{cor:LP2}, one can employ the algorithm proposed in \cite{givan2000bounded}, which has been shown to be very efficient (see also Remark \ref{rem:similarity to givan}). Intuitively, when $|\ver(\Aimdp)|>>|\Qimdp|$, then the approach of Corollary \ref{cor:LP1} might be preferred against the one of Corollary \ref{cor:LP2}, and vice-versa. Nevertheless, a thorough experimental study has to be conducted, to decide when to use each.
\begin{remark}
Max-min problems with linear objective functions and constraints, in their general form, are studied in \cite{falk1973linear}. It is therein where it is proven that an optimal solution is attained at the vertices of the constraint set. Moreover, an algorithm is proposed that decomposes the original max-min problem to a number of LPs, which in our case would be $O(|\Qimdp|^2)$ LPs. Here, we are able to derive more computationally efficient results, because we exploit our specific problem's special structure (e.g, the knowledge that $\Plow(q,a,q') \leq \Pup(q,a,q')$ and $\sum_{q' \in \Qimdp} \Plow(q,a,q') \leq 1 \leq \sum_{q' \in \Qimdp} \Pup(q,a,q')$). 
\end{remark}

\textit{\textbf{The concave/convex case.}} When $\Aimdp$ is convex, $\Plow(\cdot,a,\cdot)$ is concave and $\Pup(\cdot, a, \cdot)$ is convex on $a$, then the maximization problems in \eqref{eq:final_LP} are convex programs (CPs):
\begin{corollary}[to Theorem \ref{thm:main_theorem}]\label{cor:CP}
If $\Aimdp$ is convex, $\Plow(\cdot,a,\cdot)$ is concave on $a$ and $\Pup (\cdot,a,\cdot)$ is convex on $a$, then the maximization problems in \eqref{eq:final_LP} are convex programs (CPs). 
\end{corollary}
In such case, standard convex optimization algorithms (see \cite{boyd2004convex}) can be employed to solve \eqref{eq:final_LP} and perform value iteration over the given \caimdp.

\textit{\textbf{The convex/concave case.}} When $\Aimdp$ is a polytope, $\Plow(\cdot,a,\cdot)$ is convex and $\Pup(\cdot, a, \cdot)$ is concave on $a$, we have the following:
\begin{corollary}[to Theorem \ref{thm:main_theorem}]\label{cor:convex_max}
If $\Aimdp$ is a polytope, $\Plow(\cdot,a,\cdot)$ is convex on $a$ and $\Pup (\cdot,a,\cdot)$ is concave on $a$, then the maximization problems in \eqref{eq:final_LP} amount to maximizing a convex function over a polytope. Thus, an optimal solution lies at the constraint set's vertices.
\end{corollary}
\begin{proof}[Proof of Corollary \ref{cor:convex_max}]
The fact that maximizing a convex function over a polytope attains a solution at the vertices is well-known; e.g., see \cite[Theorem 32.2]{rockafellar1970convex}. 
\end{proof}
In this case, we can either solve all $|\Qimdp|$  problems in \eqref{eq:final_LP}, by evaluating the objective function on the vertices of $\Aimdp$, or we can solve the $|\ver(\Aimdp)|$ LPs \eqref{eq:vertices}, which, again, is equivalent to building a discrete-action IMDP, as explained in Remark \ref{rem:discreteIMDP_linear}. Either way, when $|\ver(\Aimdp)|$ is small, value iteration can be solved quickly.

\section{Numerical Example}
Here, we demonstrate this work's theoretical results on a numerical example. We consider the concave/convex case and we compare the results obtained by synthesizing over a \caimdp with the ones obtained by discrete-action IMDPs (discretized versions of the \caimdp). We constructed a randomly generated \caimdp $\I = (\Qimdp, \Aimdp, \Plow, \Pup, R)$ with $|\Qimdp| = 25$ states, concave/convex transition bounds and a 3-dimensional action space ($n_{\Aimdp}=3$). The action set is the following cylinder:
\begin{equation*}
   \Aimdp = \{a\in\real^3: (a_1-0.5)^2 + (a_2-0.5)^2 \leq 0.2, \text{ }a_3\in[0,1]\} 
\end{equation*}
The horizon $N=10$ and $\gamma = 1$.

We computed the robust optimal policy and expected cumulative reward $\mathcal{R}_{\star}(q)$ (for all $q$) over the \caimdp, via convex programming, as instructed by Corollary \ref{cor:CP}. Moreover, we obtained policies and the corresponding rewards $\mathcal{R}_s(q)$ (for all $q$) by performing synthesis over discrete-action IMDPs, which are discretized versions of the original \caimdp: a discrete-action IMDP is constructed as $\I^s_{discr}=(\Qimdp, \hat{\Aimdp}_s, \Plow, \Pup, R)$, where $\hat{\Aimdp}_s\subseteq\Aimdp$ is a finite set of $s$ actions randomly sampled from the continuous action set $\Aimdp$; that is, $|\hat{\Aimdp}_s| = s$. Specifically, we computed $\mathcal{R}_s(q)$ for different values of $s$. Note that, since randomness is involved in the experiments\footnote{In the case of value iteration over the \caimdp with convex programming, the initial condition for the optimization is picked randomly, which affects the total computation time. In the case of discrete-action IMDPs, the sampled actions affect both the computation time and, especially, the reward $\mathcal{R}_s(q)$.}, we performed them multiple times\footnote{Value iteration over the \caimdp has been performed 3 times. The experiments on discrete-action IMDPs for $s=1,8,27,64,125$ have been performed $[20, 11, 9, 6, 3]$ times, respectively.}. Moreover, to perform synthesis over the discrete-action IMDPs (i.e. to solve the $s$ associated LPs), we used the algorithm of \cite{givan2000bounded}. Finally, optimality tolerance for convex programming was set to $10^{-4}$.

\begin{figure}[h!]
    \centering
    \includegraphics[width = 0.7\textwidth]{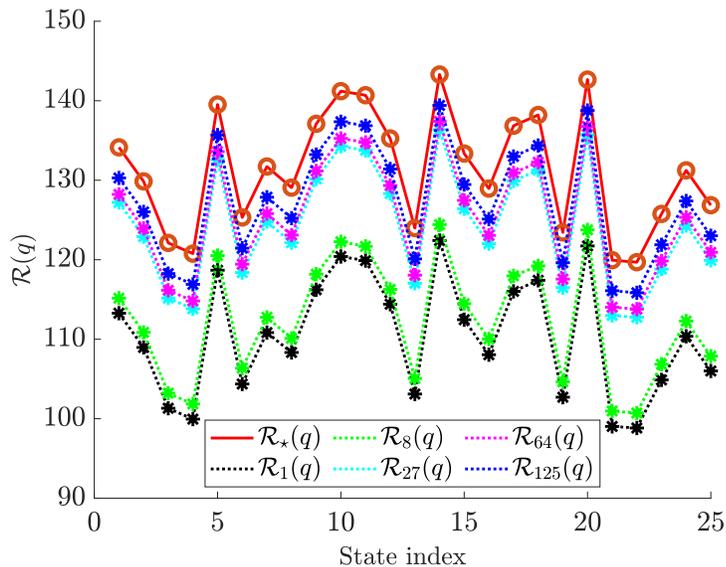}
    \caption{The optimal reward $\mathcal{R}_{\star}(q)$ (solid line) obtained from the \caimdp, and rewards $\mathcal{R}_s(q)$ (dashed lines) obtained from discrete-action IMDPs.}
    \label{fig:rewards}
\end{figure}
Figure \ref{fig:rewards} depicts the optimal reward $\mathcal{R}_{\star}(q)$ obtained from the \caimdp and rewards $\mathcal{R}_s(q)$ obtained from some of the discrete-action IMDPs, for different values of $s$. As expected, in all cases $\mathcal{R}_{\star}(q)\geq\mathcal{R}_s(q)$, since $\mathcal{R}_{\star}(q)$ is the optimal reward, whereas $\mathcal{R}_s(q)$ is suboptimal. Moreover, Table 1 shows the average percentage of suboptimality of $\mathcal{R}_s(q)$, i.e. the quantity: $$100\cdot\max_q\tfrac{\mathcal{R}_{\star}(q) - \mathcal{R}_s(q)}{\mathcal{R}_{\star}(q)}\%,$$ for all different values of $s$. Furthermore, it shows the average computation time for each experiment. We observe that performing synthesis on the \caimdp is approximately as expensive as synthesizing over a discrete-action IMDP with only 27 actions (only 3 points per dimension of the action space). That is while $\mathcal{R}_{27}(q)$ is significantly suboptimal with an average suboptimality of 5.17$\%$. In fact, even with 125 sampled actions (5 points per dimension), the suboptimality percentage is considerable. From this experiment, we deduce that synthesizing over a \caimdp is not only the way to optimality, but also it is computationally efficient even compared to discrete-action IMDPs. In fact, for even higher-dimensional action spaces, it is expected that it is the only tractable choice.
\begin{table}
\label{table:cpvsdiscrete}
\caption{The first five lines report average cpu times and average suboptimality percentages for value iteration over discrete-action IMDPs obtained by randomly sampling $s$ actions from the continuous action-space. The last line reports the corresponding results for value iteration over the \caimdp.}
\begin{tabular}{c | r | c}
\toprule
 \textbf{$\#$samples $s$} & \textbf{CPU Time (s)} &  $\boldsymbol{100\cdot\max_q\tfrac{\mathcal{R}_{\star}(q) - \mathcal{R}_s(q)}{\mathcal{R}_{\star}(q)}\%}$ \\ [0.5ex] 
 \hline
 1 &88  & 16.51$\%\;$  \\ 
 8 &668  & 7.86$\%$  \\
 27 &2170  & 5.17$\%$  \\
 64 &6225  & 4.59$\%$  \\
 125 &9736  & 4.32$\%$  \\ 
 \hline
 \textbf{\caimdp} &\textbf{2212} &\textbf{0$\boldsymbol{\%}$} \\
\bottomrule
\end{tabular}
\end{table}

\section{Discussion: \NoCaseChange{ca}IMDP\NoCaseChange{s} for control synthesis for stochastic systems}
We now, briefly, describe how this paper's results can be employed for control synthesis for stochastic systems. 

\subsection{Abstracting stochastic systems via {\caimdp}s}
Consider the controlled stochastic system:
\begin{equation}\label{eq:control_sys}
    x_{k+1} = f(x_k,a_k,w_k)
\end{equation}
where $x_k\in X$ is the state of the system at time $k$, $a_k\in\Aimdp$ is the action at time $k$ and $w_k \in \real^{n_w}$ are i.i.d random variables with a known probability distribution $p_w$. The paths of the system are again sequences $\omega = (x_1,a_1),(x_2,a_2),\dots$, and we denote the system's set of finite paths by $\Pathimdpfin_{\mathrm{sys}}$. Assume a state-dependent reward $R_{sys}:X\to \real_{\geq 0}$. Given a specific control policy $\chi:\Pathimdpfin_{\mathrm{sys}} \times X\to\Aimdp$, a horizon $N$, an initial condition $x_0\in\real^{n_x}$ and a factor $\gamma\geq 0$, the expected cumulative reward over the system's trajectories is defined as:
\begin{equation*}
    \expec^{x_0}_{\pr_w} \left[ \sum_{i=0}^N \gamma^i R_{sys}(\omega^q(i)) \mid \chi \right]
\end{equation*}
where $\pr_w$ is a probability measure over the system's paths induced by $p_w$ and the initial condition $x_0$. This expected reward constitutes a metric describing the system's behaviors under policy $\chi$ (e.g., a reachability or safety property, a quantitative measure such as convergence speed, etc.). The objective is to synthesize a policy $\chi_\star$ that maximizes the expected reward.

In the literature of IMDP-abstractions for control synthesis for stochastic systems \citep{lahijanian2015formal,laurenti2020formal,coogan2022abstraction}, the above problem is solved by abstracting the system via an appropriate IMDP, finding an optimal policy over the constructed IMDP and mapping this policy back to the original system. Nevertheless, as already mentioned in the introduction section, the action set $\Aimdp$ is, generally, assumed discrete a-priori. To the best of our knowledge, the only exception is \cite{coogan2022abstraction}, which keeps the action-space continuous, but the max-min problem of value iteration is solved using heuristics.

With the results derived in this work, we can avoid unnecessary discretizations of the action set or heuristics with no formal guarantees. To abstract a system \eqref{eq:control_sys} by a \caimdp, the state space $X$ is partitioned into $|\Qimdp|$ sets $q_i \subseteq X$, which represent the {\caimdp}'s states, and the transition bounds are defined as follows for all $q_i,q_j \in \Qimdp$:
\begin{align}\label{eq:transition_bounds}
    &\Plow(q_i,a,q_j) \leq \inf\limits_{x\in q_i}\pr_w(f(x,a,w) \in q_j), \\&\Pup(q_i,a,q_j) \geq \sup\limits_{x\in q_i}\pr_w(f(x,a,w) \in q_j)
\end{align}
The \caimdp-reward is defined by $R(q) = \inf_{x\in q}R_{sys}(q)$. By extending \cite[Theorem 4.1]{delimpaltadakis2022formal} to {\caimdp}s, we can show that the optimal expected cumulative {\caimdp}-reward lower-bounds the optimal expected cumulative reward over the system's trajectories:
\begin{equation}
    \max_{\str}\min_{\adv}\mathcal{R}^{N}_{\str,\adv}(q_0)\leq \max_{\chi}\expec^{x_0}_{\pr_w}\left[\sum_{i=0}^N \gamma^i R_{sys}(\omega^q(i)) \mid \chi \right]
\end{equation}
where $x_0\in q_0$. Furthermore, the optimal reward obtained via the \caimdp is larger than or equal to an optimal reward obtained via a discrete-action IMDP, as its action set is a subset of the {\caimdp}'s action set.

In order to enable the results presented in Section \ref{sec:cases} and render value iteration easy-to-solve, the bounds $\Plow(\cdot,a,\cdot)$ and $\Pup(\cdot,a,\cdot)$ have to be (piecewise) linear or concave/convex or convex/concave on the action variables. For instance, this is the case for a large class of switching diffusion models where the continuous action controls the probability of switching between different modes of the system \citep{yin2009hybrid}. In more general cases, the functions $\inf_{x\in q_i}$ and $\sup_{x\in q_i}\pr_w(f(x,a,w) \in q_j)$ may need to be lower-/upper-bounded via (piecewise) linear, convex or concave functions.

\subsection{Suboptimality bounds}
A policy derived through a {\caimdp}-abstraction is generally suboptimal for system \eqref{eq:control_sys}, even though the algorithms presented here do compute a policy that is optimal w.r.t. the given {\caimdp}. That is because the \caimdp is an abstraction of the original stochastic system in \eqref{eq:control_sys}, and not an equivalent representation. Nevertheless, {\caimdp}s can be used to obtain so-called \emph{suboptimality bounds}: bounds on the difference between the optimal reward computed via the \caimdp and the true optimal reward:
\begin{equation*}
     \max_{\chi}\expec^{x_0}_{\pr_w}\left[\sum_{i=0}^N \gamma^i R_{sys}(\omega^q(i)) \mid \chi \right] - \max_{\str}\min_{\adv}\mathcal{R}^{N}_{\str,\adv}(q_0)
\end{equation*}
This quantifies how suboptimal a policy obtained via a \caimdp is; if the obtained result is not satisfactory, then the {\caimdp}-abstraction has to be refined with a finer partition of the state-space $X$ or with tighter transition bounds.

Again, by extending \cite[Theorem 4.1]{delimpaltadakis2022formal} to {\caimdp}s, it can be shown that, if we define the reward function of the \caimdp as $R'(q) = \sup_{x\in q}R_{sys}(q)$ (previously, we used $\inf$), the \emph{optimistically optimal} \caimdp reward upper bounds the true optimal one:
\begin{equation*}
    \max_{\chi}\expec^{x_0}_{\pr_w} \left[ \sum_{i=0}^N \gamma^i R_{sys}(\omega^q(i)) \mid \chi \right] \leq \max_{\str,\adv}\mathcal{R'}^{N}_{\str,\adv}(q_0)
\end{equation*}
where $\mathcal{R'}^{N}_{\str,\adv}(q_0)=\expec^{q_0}\left[\sum_{i=0}^N \gamma^i R'(\omega^q(i)) \mid \str,\adv \right]$, and the expectation is w.r.t. the ptobability measure over the paths of the \caimdp. Thus, $\max_{\str,\adv}\mathcal{R'}^{N}_{\str,\adv}(q_0) - \max_{\str}\min_{\adv}\mathcal{R}^{N}_{\str,\adv}(q_0)$ provides a suboptimality bound. 

The above requires solving $\max_{\str,\adv}\mathcal{R}^{N}_{\str,\adv}(q_0)$, which amounts to a value iteration, in which, instead of a max-min problem, the following max problem needs to be solved:
\begin{equation}\label{eq:max-max}
\begin{aligned}
	\max_{a \in \Aimdp, p \in \real^{|\Qimdp|}} \quad & p^\top V_{k}\\
	\mathrm{s.t.:} \quad & \check{P}(q,a,q_i)-p_i \leq 0, \text{ }i=1,2,\dots,|\Qimdp|,\\
				  & -\hat{P}(q,a,q_i)+p_i \leq 0, \text{ }i=1,2,\dots,|\Qimdp|,\\
				  & \sum_{i=1}^{|\Qimdp|}p_i -1 = 0\\
\end{aligned}
\end{equation}
Referring to the three cases distinguished in Section \ref{sec:cases}, problem \eqref{eq:max-max}: a) in the linear case, is also an LP, b) in the convex/concave case, is a CP, but c) in the concave/convex case, its constraint set is nonconvex: the intersection of sub-level sets of convex and concave functions. Hence, in the linear and the convex/concave case, problem \eqref{eq:max-max} can be solved efficiently, to generate a tight suboptimality bound. In the concave/convex case, we could derive an upper bound on the optimal value of problem \eqref{eq:max-max}, e.g., by using an SMT solver \cite{dreal}.


\section{Conlusion}
We have introduced continuous-action interval Markov decision processes ({\caimdp}s), and studied value iteration for optimizing pessimistic expected cumulative rewards. Specifically, we have shown that the max-min problem associated to value iteration can be decomposed to $|\Qimdp|$ maximization problems. The simple form of these max problems allowed us to distinguish cases where they, and thus value iteration, can be solved efficiently. It also provided us with further interesting insights, such as cases where synthesis over discrete-action and continuous-action IMDPs is equivalent. These results have been demonstrated on a numerical example. 

Furthermore, we have, briefly, discussed how {\caimdp}s can be employed for control synthesis for stochastic systems. Specifically, we have given guidelines on how to construct the {\caimdp}-abstraction such that it falls in one of the aforementioned easy-to-solve cases. We have, also, shown how to obtain suboptimality bounds on the policy generated by the {\caimdp}-abstraction.

Future work includes: i) studying convergence of the value iteration for infinite horizons, ii) addressing rewards that depend both on the state and the action, iii) investigating if the $|\Qimdp|$ maximization problems can be simplified even further, and iv) a thorough study on {\caimdp}-abstractions and on obtaining suboptimality bounds.

\section{Technical proofs}\label{sec:proofs}
\subsection{Proof of Proposition \ref{prop:caimdps_value function}}\label{sec:appendix_proof_value_iteration}
\begin{proof}[Sketch of Proof of Proposition \ref{prop:caimdps_value function}]
By the transformation presented in \citep[Appendix A]{nilim2005robust} for finite-action IMDPs, the robust control problem can be equivalently transformed into a turn-based zero-sum game. Consequently, similarly to \cite{nilim2005robust}, we can employ the same reasoning as in the proof of \cite[Theorem 1]{nowak1984zero}, to show  that:
$$\mathcal{R}^{j}_\star(q_0) = V_{N-j}(q_0)$$
where $V_{N-j}$ comes from the following Bellman recursion:
\begin{align*}
    &V_{N}(q) = R(q)\label{eq:original_vi1}\\
    &V_{k-1}(q) = R(q) + \gamma\max_{\mu\in M}\int_\Aimdp \min_{p\in\Gamma_{q,a}}\sum_{q'\in\Qimdp}p(q')V_{k}(q') \mu(a)da(q')
\end{align*}
where $M$ is the set of probability distributions\footnote{Maximization takes place over the set $M$ of probability distributions on $\Aimdp$, because, as mentioned in Remark \ref{rem:deterministic policies}, policies can be probabilistic. Nonetheless, it is shown that optimal policies are deterministic (or, in other words, Dirac distributions).} over $\Aimdp$.  We now notice that for any $\mu\in M$:
$$ \int_\Aimdp \min_{p\in\Gamma_{q,a}}\sum_{q'\in\Qimdp}p(q')V_{k}(q') \mu(a)da \leq \max_{a\in \Aimdp}\min_{p\in\Gamma_{q,a}}\sum_{q'\in\Qimdp}p(q')V_{k}(q')$$
This implies that the optimal choice is a Dirac distribution centered in any $a\in \argmax_{a\in\Aimdp}\min_{p\in\Gamma_{q,a}}\sum_{q'\in\Qimdp}p(q')V_{k}(q')$. Consequently, we can replace the maximization over $M$ with a maximization over $\Aimdp,$ i.e.:
\begin{align*}
    \max_{\mu\in M}\int_\Aimdp \min_{p\in\Gamma_{q,a}}\sum_{q'\in\Qimdp}p(q')V_{k}(q') \mu(a)da = \max_{a\in \Aimdp}\min_{p\in\Gamma_{q,a}}\sum_{q'\in\Qimdp}p(q')V_{k}(q')
\end{align*}
\end{proof}

\subsection{Proof of Proposition \ref{prop:dual_problem}}\label{sec:appendix_proof_duality_prop}
First, let us prove the following lemma:
\begin{lemma}\label{duality_lemma}Consider the following max-min problem:
\begin{equation*}
	\begin{aligned}
		\max_{a \in \Aimdp} \min_{p \in P, f(a,p) \leq 0}  \quad &c^\top p
	\end{aligned}
\end{equation*}
Assume that it is feasible and bounded, and that strong duality holds for the inner minimization problem and its Lagrangian dual (see, e.g., \cite{boyd2004convex}):
\begin{equation}\label{lemma_duality_eq}
	\forall a\in\Aimdp: \quad \min_{p \in P, f(a,p) \leq 0}  c^\top p = \max_{\lambda\succeq0, \nu}g(a,\lambda,\nu)
\end{equation}
Then $\max\limits_{a \in \Aimdp} \min\limits_{p \in P, f(a,p) \leq 0}  c^\top p = \max\limits_{a\in\Aimdp,\lambda\succeq0, \nu}g(a,\lambda,\nu)$. Moreover:
\begin{equation*}
    (a_\star, \lambda_\star,\nu_\star)\in\argmax\limits_{a\in\Aimdp,\lambda\succeq 0, \nu}g(a,\lambda,\nu)\implies a_\star\hspace{-1mm}\in\argmax\limits_{a\in\Aimdp}\min\limits_{p \in P, f(a,p)\leq 0}\hspace{-4mm}c^\top p
\end{equation*}
\end{lemma}

\begin{proof}[Proof of Lemma \ref{duality_lemma}]
We have that:
\begin{equation}\label{lemma_proof_eq1}
	\forall a\in\Aimdp: \quad  g(a_\star,\lambda_\star,\nu_\star) \geq \max_{\lambda\succeq0, \nu}g(a,\lambda,\nu) = \min_{p \in P, f(a,p) \leq 0}  c^\top p
\end{equation}
where, for the inequality we used the definition of $(a_\star, \lambda_\star,\nu_\star)$ and for the equality we used \eqref{lemma_duality_eq}. Moreover, we have:
\begin{equation}\label{lemma_proof_eq2}
	g(a_\star,\lambda_\star,\nu_\star) = \max_{\lambda\succeq0, \nu}g(a_\star,\lambda,\nu) = \min_{p \in P, f(a_\star,p) \leq 0}  c^\top p
\end{equation}
where for the first equality we used the definition of $(a_\star, \lambda_\star,\nu_\star)$ and for the second one we used \eqref{lemma_duality_eq}. Replacing \eqref{lemma_proof_eq2} into \eqref{lemma_proof_eq1} we get:
\begin{equation*}
	\forall a\in\Aimdp: \quad \min_{p \in P, f(a_\star,p) \leq 0}  c^\top p \geq \min_{p \in P, f(a,p) \leq 0}  c^\top p
\end{equation*}
Since the above holds for all $a\in\Aimdp$ and since $a_{\star}\in \Aimdp$, then:
\begin{equation*}
	\min_{p \in P, f(a_\star,p) \leq 0}  c^\top p = \max_{a\in\Aimdp}\min_{p \in P, f(a,p) \leq 0}  c^\top p
\end{equation*}
which, since the left-hand side is equal to $g(a_{\star},\lambda_{\star},\nu_{\star})$, proves that $\max\limits_{a \in A} \min\limits_{p \in P, f(a,p) \leq 0}  c^\top p = \max\limits_{a\in\Aimdp,\lambda\succeq0, \nu}g(a,\lambda,\nu)$. Moreover, it implies that $a_\star\in\argmax\limits_{a\in\Aimdp}\min\limits_{p \in P, f(a,p)\leq 0}c^\top p$. The proof is complete.
\end{proof}
Now, we are ready to prove Proposition \ref{prop:dual_problem}.
\begin{proof}[Proof of Proposition \ref{prop:dual_problem}]
The Lagrange dual function that corresponds to the inner minimization problem of \eqref{original_opti_problem} is:
\begin{align*}
	g(a,\lambda_L,\lambda_U, \nu) &= \inf_{p \in \real^{|\Qimdp|}}\Big[p^\top V_{k} + \sum_{i=1}^{|\Qimdp|}\lambda_{L_i}\Big(\check{P}(q,a,q_i)-p_i\Big)++\sum_{i=1}^{|\Qimdp|}\lambda_{U_i}\Big(-\hat{P}(q,a,q_i)+p_i\Big)+\nu\sum_{i=1}^{|\Qimdp|}p_i - \nu\Big]\\
	&=\sum_{i=1}^{|\Qimdp|}\lambda_{L_i}\check{P}(q,a,q_i) - \sum_{i=1}^{|\Qimdp|}\lambda_{U_i}\hat{P}(q,a,q_i) -\nu + \inf_{p \in \real^{|\Qimdp|}}\Big[p^T(V_{k}-\lambda_L + \lambda_U + \nu\ind)\Big]
\end{align*}
For any $a\in\Aimdp$, the Lagrangian dual problem of the minimization problem in \eqref{original_opti_problem} is (see \cite{boyd2004convex}):
\begin{equation}\label{eq:duality_prop_proof_1}
	\begin{aligned}
		\max_{\lambda_L, \lambda_U, \nu} \quad &g(a,\lambda_L,\lambda_U, \nu)\\
		\mathrm{s.t.:} \quad &\lambda_L \succeq 0, \text{ }\lambda_U \succeq 0, \text{ }\nu \in\real
	\end{aligned}
\end{equation}
As done in \cite{boyd2004convex}, to maximize $g(a,\lambda_L,\lambda_U, \nu)$, we have to make $V_{k}-\lambda_L + \lambda_U + \nu\ind=0$, or otherwise the term $\inf_{p \in \real^{|\Qimdp|}}\Big[p^T(V_{k}-\lambda_L + \lambda_U + \nu\ind)\Big]$ will always be $-\infty$. Thus, \eqref{eq:duality_prop_proof_1} becomes:
\begin{equation}\label{eq:duality_prop_proof_2}
	\begin{aligned}
		\max_{\lambda_L, \lambda_U, \nu} \quad &\sum_{i=1}^{|\Qimdp|}\lambda_{L_i}\check{P}(q,a,q_i) - \sum_{i=1}^{|\Qimdp|}\lambda_{U_i}\hat{P}(q,a,q_i) -\nu\\
		\mathrm{s.t.:} \quad &\text{ }\lambda_L \succeq 0, \text{ }\lambda_U \succeq 0, \text{ }\nu \in\real,\text{ }V_{k}-\lambda_L + \lambda_U + \nu\ind=0
	\end{aligned}
\end{equation}
As the inner minimization problem of \eqref{original_opti_problem} is an LP, strong duality holds between the minimization problem of \eqref{original_opti_problem} and \eqref{eq:duality_prop_proof_2}. Thus, we can apply Lemma \eqref{duality_lemma}, and the proof is completed.
\end{proof}

\subsection{Proof of Theorem \ref{thm:main_theorem}}\label{sec:appendix_proof_main}
\begin{proof}[Proof of Theorem \ref{thm:main_theorem}]
From Proposition \ref{prop:dual_problem}, we only need to solve \eqref{eq:dual_opti_problem}, in order to calculate the optimal value of \eqref{original_opti_problem} and to find an optimal action $a_{\star}$ that solves \eqref{original_opti_problem}. Thus, for the following, we focus on \eqref{eq:dual_opti_problem}.

\textit{Eliminating the $\lambda$-variables.} First, we eliminate the equality constraint, and thus the variable $\lambda_L$, by replacing $\lambda_L = V_{k} + \lambda_U + \nu\ind$ in the objective function and all other constraints:
\begin{equation}\label{eq:dual_eliminated_equality}
	\begin{aligned}
		\max_{a, \lambda_U, \nu} \quad &\sum_{i=1}^{|\Qimdp|}(V_{k}(q_i) + \lambda_{U_i} + \nu)\check{P}(q,a,q_i) - \sum_{i=1}^{|\Qimdp|}\lambda_{U_i}\hat{P}(q,a,q_i) -\nu=\\
		\max_{a, \lambda_U, \nu} \quad &\sum_{i=1}^{|\Qimdp|}V_{k}(q_i)\check{P}(q,a,q_i)  + \sum_{i=1}^{|\Qimdp|}\lambda_{U_i}\underbrace{(\check{P}(q,a,q_i)-\hat{P}(q,a,q_i))}_{\leq 0}+\nu\underbrace{(\sum_{i=1}^{|\Qimdp|}\check{P}(q,a,q_i)-1)}_{\leq 0}\\
		\mathrm{s.t.:} \quad & a\in\Aimdp, \text{ }\lambda_U \succeq 0, \text{ }\nu \in\real,\text{ }V_{k} + \lambda_U + \nu\ind \succeq 0
	\end{aligned}
\end{equation}
For any given $a$ and $\nu$, since $(\check{P}(q,a,q_i)-\hat{P}(q,a,q_i))\leq 0$ and $\lambda_U\succeq 0$, the variables $\lambda_{U_i}$ have to be made as small as possible, in order to maximize the objective function. Due to the two constraints $\lambda_U\succeq 0$ and $V_{k} + \lambda_U + \nu\ind \succeq 0$, this is encoded as	$\lambda_{U_i} = \max\Big(0, - V_{k}(q_i)-\nu\Big)$. Thus, optimization problem \eqref{eq:dual_eliminated_equality} becomes:
\begin{equation}\label{eq:dual_eliminated_lambdas}
	\begin{aligned}
		\max_{a, \nu} \quad &\sum_{i=1}^{|\Qimdp|}V_{k}(q_i)\check{P}(q,a,q_i)  + \sum_{i=1}^{|\Qimdp|}\Big[\max\Big(0, -V_{k}(q_i)-\nu\Big)\underbrace{(\check{P}(q,a,q_i)-\hat{P}(q,a,q_i))}_{\leq 0}\Big] +\nu\underbrace{(\sum_{i=1}^{|\Qimdp|}\check{P}(q,a,q_i)-1)}_{\leq 0}\\
		\mathrm{s.t.:} \quad &a\in\Aimdp, \text{ }\nu \in\real
	\end{aligned}
\end{equation}
where both constraints $\lambda_U\succeq 0$ and $V_{k} + \lambda_U + \nu\ind \succeq 0$ have been by-construction eliminated by fixing $\lambda_{U_i} = \max\Big(0, -V_{k}(q_i)-\nu\Big)$. We are now left only with the optimization variables $a$ and $\nu$.

\textit{Deriving $|\Qimdp|$ simpler max problems from \eqref{eq:dual_eliminated_lambdas}.} Problem \eqref{eq:dual_eliminated_lambdas} contains the product $\max\Big(0, -V_{k}(q_i)-\nu\Big)(\check{P}(q,a,q_i)-\hat{P}(q,a,q_i))$, which is hard to analyze directly. Let us further simplify it. Since $V_{k}$ is in descending order, optimization problem \eqref{eq:dual_eliminated_lambdas} can be broken down to $|\Qimdp|+1$ simpler problems, as follows: 
\begin{enumerate}
	\item The 1st optimization problem is the same with \eqref{eq:dual_eliminated_lambdas}, but $\nu$ is restricted by the inequality constraint $\nu \leq -V_{k}(q_{1})$.
	\item The $(|\Qimdp|+1)$-th optimization problem is the same with \eqref{eq:dual_eliminated_lambdas}, but $\nu$ is restricted by the inequality constraint $\nu \geq -V_{k}(q_{|\Qimdp|})$.
	\item For $j\in \{2,\dots, |\Qimdp|\}$, the $j$-th optimization problem is the same with \eqref{eq:dual_eliminated_lambdas}, but $\nu$ is restricted by the inequality constraints $-V_{k}(q_{j-1}) \leq \nu \leq - V_{k}(q_{j})$.
\end{enumerate}
Let us examine each case separately.

\textbf{Case 1:} Since $\nu \leq -V_{k}(q_{1})$, then $\nu \leq -V_{k}(q_{i})$ for all $i=1,2,\dots,|\Qimdp|$. Thus, $\max\Big(0, -V_{k}(q_i)-\nu\Big)=-V_{k}(q_i)-\nu$ for all $i$. Hence, optimization problem \eqref{eq:dual_eliminated_lambdas} becomes:
\begin{equation*}
	\begin{aligned}
		\max_{a, \nu} \quad &\cancel{\sum_{i=1}^{|\Qimdp|}V_{k}(q_i)\check{P}(q,a,q_i)}  + \sum_{i=1}^{|\Qimdp|}\Big[\Big( -V_{k}(q_i)-\nu\Big)(\cancel{\check{P}(q,a,q_i)}-\hat{P}(q,a,q_i))\Big] +\nu(\cancel{\sum_{i=1}^{|\Qimdp|}\check{P}(q,a,q_i)}-1)=\\
		\max_{a, \nu} \quad &\sum_{i=1}^{|\Qimdp|}V_{k}(q_i)\hat{P}(q,a,q_i) +\nu\underbrace{(\sum_{i=1}^{|\Qimdp|}\hat{P}(q,a,q_i)-1)}_{\geq 0}\\
		\mathrm{s.t.:} \quad & a\in\Aimdp, \text{ }\nu \leq -V_{k}(q_{1})
	\end{aligned}
\end{equation*}
Now, to maximize the objective function, given any $a$, $\nu$ has to be made as big as possible, i.e. $\nu = -V_{k}(q_{1})$. Hence, the optimization problem transforms to the following:
\begin{equation}\label{LP1} 
	\begin{aligned}
		\max_{a} \quad &\sum_{i=1}^{|\Qimdp|}\Big(V_{k}(q_i)-V_{k}(q_{1})\Big)\hat{P}(q,a,q_i)+V_{k}(q_{1})=\\
		\max_{a} \quad &\sum_{i=2}^{|\Qimdp|}\Big(V_{k}(q_i)-V_{k}(q_{1})\Big)\hat{P}(q,a,q_i)+V_{k}(q_{1})\\
		\mathrm{s.t.:} \quad & a\in\Aimdp
	\end{aligned}
\end{equation}

\textbf{Case 2:} Since $\nu \geq -V_{k}(q_{|\Qimdp|})$, then $\nu \geq -V_{k}(q_{i})$ for all $i=1,2,\dots,|\Qimdp|$. Thus, $\max\Big(0, -V_{k}(q_i)-\nu\Big)=0$ for all $i$. Hence, optimization problem \eqref{eq:dual_eliminated_lambdas} becomes:
\begin{equation*}
	\begin{aligned}
		\max_{a, \nu} \quad &\sum_{i=1}^{|\Qimdp|}V_{k}(q_i)\check{P}(q,a,q_i)  +\nu\underbrace{(\sum_{i=1}^{|\Qimdp|}\check{P}(q,a,q_i)-1)}_{\leq 0}\\
		\mathrm{s.t.:} \quad & a\in\Aimdp, \text{ }\nu \geq -V_{k}(q_{|\Qimdp|})
	\end{aligned}
\end{equation*}
To maximize the objective function, given any $a$, $\nu$ has to be made as small as possible, i.e. $\nu = -V_{k}(q_{|\Qimdp|})$. Hence, the optimization problem transforms to the following:
\begin{equation}\label{LP2}
	\begin{aligned}
		\max_{a} \quad &\sum_{i=1}^{|\Qimdp|}\Big(V_{k}(q_i) - V_{k}(q_{|\Qimdp|})\Big)\check{P}(q,a,q_i)  +V_{k}(q_{|\Qimdp|})=\\
		\max_{a} \quad &\sum_{i=1}^{|\Qimdp|-1}\Big(V_{k}(q_i) - V_{k}(q_{|\Qimdp|})\Big)\check{P}(q,a,q_i)  +V_{k}(q_{|\Qimdp|})=\\
		\mathrm{s.t.:} \quad & a\in\Aimdp
	\end{aligned}
\end{equation}

\textbf{Case 3:} 
In this case, we have that:
\begin{equation*}
	\max\Big(0, -V_{k}(q_i)-\nu\Big)=\left\{\begin{aligned}
		0, \quad &i=1,2,\dots,j-1\\
		-V_{k}(q_i)-\nu, \quad &i=j,\dots,|\Qimdp|
	\end{aligned}\right.
\end{equation*}
Then, optimization problem \eqref{eq:dual_eliminated_lambdas} becomes:
\begin{equation*}
	\begin{aligned}
		\max_{a, \nu} \quad &\sum_{i=1}^{|\Qimdp|}V_{k}(q_i)\check{P}(q,a,q_i)  + \sum_{i=j}^{|\Qimdp|}\Big[\Big( -V_{k}(q_i)-\nu\Big)(\check{P}(q,a,q_i)-\hat{P}(q,a,q_i))\Big] +\nu(\sum_{i=1}^{|\Qimdp|}\check{P}(q,a,q_i)-1)=\\
		\max_{a, \nu} \quad &\sum_{i=1}^{j-1}V_{k}(q_i)\check{P}(q,a,q_i)  + \sum_{i=j}^{|\Qimdp|}V_{k}(q_i)\hat{P}(q,a,q_i) + \nu\Big(\sum_{i=1}^{j-1}\check{P}(q,a,q_i) + \sum_{i=j}^{|\Qimdp|}\hat{P}(q,a,q_i) - 1\Big)\\
		\mathrm{s.t.:} \quad & a\in\Aimdp, \text{ }-V_{k}(q_{j-1}) \leq \nu \leq -V_{k}(q_{j})
	\end{aligned}
\end{equation*}
Now, for any given $a$, the term $\sum_{i=1}^{j-1}\check{P}(q,a,q_i) \allowdisplaybreaks+ \sum_{i=j}^{|\Qimdp|}\hat{P}(q,a,q_i) - 1$ is either negative or positive. In the first case, we would have to make $\nu$ as small as possible, i.e. $\nu = -V_{k}(q_{j-1})$. In the second case, we would have to make $\nu$ as big as possible, i.e. $\nu = -V_{k}(q_{j})$. Thus, we can break down the above optimization problem to the two following problems \eqref{LP3} and \eqref{LP4}:
\begin{equation}\label{LP3} 
	\begin{aligned}
		\max_{a} \quad &\sum_{i=1}^{j-1}\Big(V_{k}(q_i)-V_{k}(q_{j-1})\Big)\check{P}(q,a,q_i)  + \sum_{i=j}^{|\Qimdp|}\Big(V_{k}(q_i)-V_{k}(q_{j-1})\Big)\hat{P}(q,a,q_i) +V_{k}(q_{j-1})=\\
		\max_{a} \quad&\sum_{i=1}^{j-2}\Big(V_{k}(q_i)-V_{k}(q_{j-1})\Big)\check{P}(q,a,q_i)  + \sum_{i=j}^{|\Qimdp|}\Big(V_{k}(q_i)-V_{k}(q_{j-1})\Big)\hat{P}(q,a,q_i) +V_{k}(q_{j-1})\\
		\mathrm{s.t.:} \quad & a\in\Aimdp
	\end{aligned}
\end{equation}
\begin{equation}\label{LP4}
	\begin{aligned}
		\max_{a} \quad &\sum_{i=1}^{j-1}\Big(V_{k}(q_i)-V_{k}(q_{j})\Big)\check{P}(q,a,q_i)  + \sum_{i=j}^{|\Qimdp|}\Big(V_{k}(q_i)-V_{k}(q_{j})\Big)\hat{P}(q,a,q_i) +V_{k}(q_{j})=\\
		\max_{a} \quad&\sum_{i=1}^{j-1}\Big(V_{k}(q_i)-V_{k}(q_{j})\Big)\check{P}(q,a,q_i)  + \sum_{i=j+1}^{|\Qimdp|}\Big(V_{k}(q_i)-V_{k}(q_{j})\Big)\hat{P}(q,a,q_i) +V_{k}(q_{j})\\
		\mathrm{s.t.:} \quad & a\in\Aimdp
	\end{aligned}
\end{equation}

Thus far, we have that \eqref{eq:dual_opti_problem}, and thus the primal problem \eqref{original_opti_problem}, has been broken down to the combination of problem \eqref{LP1}, problem \eqref{LP2} and $|\Qimdp|-1$ instances of problems \eqref{LP3} and \eqref{LP4}. However, taking a closer look, we see that \eqref{LP4} with $j=m$ is identical to \eqref{LP3} with $j=m+1$ (their objective functions are identical, as well as their constraint sets)\footnote{Denote by $f_{j}(a)$ the objective function of \eqref{LP4} and by $g_{j}(a)$ the objective function of \eqref{LP3}. Then, we have:
\begin{align*}
	f_{m}(a) - g_{m+1}(a) &= \hspace{-2mm}\sum_{i=1}^{m-1}\Big(V_{k}(q_i)-V_{k}(q_{m})\Big)\check{P}(q,a,q_i) - \sum_{i=1}^{m-1}\Big(V_{k}(q_i)-V_{k}(q_{m})\Big)\check{P}(q,a,q_i)\\
	&+\hspace{-2mm}\sum_{i=m+1}^{|\Qimdp|}\Big(V_{k}(q_i)-V_{k}(q_{m})\Big)\hat{P}(q,a,q_i) - \hspace{-2mm}\sum_{i=m+1}^{|\Qimdp|}\Big(V_{k}(q_i)-V_{k}(q_{m})\Big)\hat{P}(q,a,q_i)\\ &+V_{k}(q_{m})-V_{k}(q_{m})\\
	&=0
\end{align*}}. Moreover, \eqref{LP2} is identical to \eqref{LP4} with $j=|\Qimdp|$. Finally, \eqref{LP1} is the same problem as \eqref{LP3} with $j=2$, which can also be written in the form of \eqref{LP4} if we put $j=1$. In other words, it suffices to only consider \eqref{LP4} for $j=1,\dots,|\Qimdp|$. To conclude, solving \eqref{eq:dual_opti_problem} is equivalent to solving the following optimization problems for $j=1,\dots,|\Qimdp|$:
\begin{equation*}
	\begin{aligned}
		\max_{a} \quad &\sum_{i=1}^{j-1}\Big(V_{k}(q_i)-V_{k}(q_{j})\Big)\check{P}(q,a,q_i)  + \sum_{i=j+1}^{|\Qimdp|}\Big(V_{k}(q_i)-V_{k}(q_{j})\Big)\hat{P}(q,a,q_i) +V_{k}(q_{j})\\
		\mathrm{s.t.:} \quad &a\in\Aimdp
	\end{aligned}
\end{equation*}
\end{proof}
\vfill
\begin{acks}
The authors would like to thank Dr. Sergio Grammatico, Dr. Amin Sharifi Kolarijani, and Dr. Gabriel de Albuquerque Gleizer for their helpful comments. 

This work was supported in part by the \grantsponsor{GSERC}{European Research Council}{https://erc.europa.eu/} through the SENTIENT project, Grant No.~\grantnum[https://cordis.europa.eu/project/id/755953]{GSERC}{ERC-2017-STG \#755953}, and \grantsponsor{NSF}{National Science Foundation (NSF)}{https://nsf.gov/} award 2039062.
\end{acks}
\bibliographystyle{ACM-Reference-Format}
\bibliography{refs}
\end{document}